\documentclass[10pt]{amsart}

\usepackage{amsmath,amssymb,amsthm,amscd,amsfonts}
\usepackage{fullpage}
\usepackage{graphicx}
\usepackage{stmaryrd}
\usepackage{float}
\usepackage{hyperref}

\usepackage{algorithm}							
\usepackage[noend]{algpseudocode}					
\usepackage{caption}
\usepackage[margin=1in]{geometry}
\usepackage{url}
\usepackage{xcolor}

\usepackage{mathtools}

\usepackage{amssymb,latexsym}
\usepackage{amsmath,amsfonts,amsthm,amscd,amsxtra,enumerate}
\usepackage{mathtools}
\usepackage{caption}
\usepackage[all]{xy}

\usepackage[margin=1in]{geometry}
\usepackage{epsfig}
\usepackage{pgf,tikz}

\usepackage[utf8]{inputenc}

\newtheorem{conjecture}{ Conjecture}[section]
\newtheorem{theorem}[conjecture]{ Theorem}
\newtheorem{lemma}[conjecture]{ Lemma}
\newtheorem{corollary}[conjecture]{ Corollary}
\newtheorem{proposition}[conjecture]{ Proposition}
\theoremstyle{definition}

\newtheorem{example}[conjecture]{ Example}

\begin{document}

\title{Theta palindromes in theta conjugates}

\author{Kalpana Mahalingam, Palak Pandoh, Anuran Maity}

\address
	{Department of Mathematics,\\ 
	 Indian Institute of Technology Madras, 
	  Chennai, 600036, India}
	  \email{kmahalingam@iitm.ac.in,palakpandohiitmadras@gmail.com, anuran.maity@gmail.com}
	

%
%

\begin{abstract}
A DNA string is a Watson-Crick (WK) palindrome when the complement of its reverse is equal to itself. The Watson-Crick mapping $\theta$ is an involution that is also an antimorphism. $\theta$-conjugates of a word is a generalisation of conjugates of a word that incorporates the notion of WK-involution $\theta$. In this paper,
we study the distribution of palindromes and Watson-Crick palindromes, also known as $\theta$-palindromes among both the set of conjugates and $\theta$-conjugates of a word $w$. We also consider some general properties of the set $C_{\theta}(w)$, i.e., the set of $\theta$-conjugates of a word $w$, and characterize words $w$ such that $|C_{\theta}(w)|=|w|+1$, i.e., with the maximum number of elements in  $C_{\theta}(w)$. We also find the structure of words that have at least one (WK)-palindrome in $C_{\theta}(w)$.
\end{abstract}
\maketitle

\section{Introduction}
The study of sequences have applications in numerous fields such as biology, computer science, mathematics, and physics. DNA  molecules, which  carry the genetic information in almost all organisms, play an important role in molecular biology (see \cite{s2,s1,s3,s4}).
DNA computing experiments use 
information-encoding strings that possess  Watson-Crick complementarity property  
between DNA single-strands which allows information-encoding strands to potentially
interact. Formally, the Watson-Crick complementarity property on strings over $\Sigma$ is  an involution $\theta$  with
the additional property that $\theta(uv)=\theta(v)\theta(u)$ for all strings $u,v \in  \Sigma^*$ where $\theta$ is an involution, i.e., $\theta^2$ equals the identity.

The notion of $\theta$-palindrome was defined in \cite{watson} to study  palindromes from the perspective of DNA computing. It was defined
independently in \cite{LUCA2006282},  where closure operators for $\theta$-palindromes were considered. The classical results on  conjugacy and commutativity of words are present in \cite{Shyr}.  In \cite{CZEIZLER2010617}, the authors
study the properties of $\theta$-primitive
words. They prove
the existence of a unique $\theta$-primitive root of a given word, and provided some constraints
under which two distinct words  share their $\theta$-primitive root.  The combinatorial properties of strings in connection to partial words were investigated in \cite{BLANCHETSADRI2002297}. The  notions of conjugacy
and commutativity was generalized to incorporate the notion of
Watson-Crick complementarity of DNA
single-strands in \cite{watson}. The authors define and study properties of Watson-Crick conjugate
and commutative words, as well as Watson-Crick palindromes. They provide a complete characterization of the set of all words
that are not Watson-Crick palindromes. Some properties that link the
Watson-Crick palindromes to classical notions such as that of primitive words are established in \cite{Kari2010}. The authors show
that the set of Watson-Crick-palindromic words that cannot be written as the product of two non-empty
Watson-Crick-palindromes equals the set of primitive Watson-Crick-palindromes.

In this paper, we extend the notion  of palindromes in conjugacy class of a word to  Watson-Crick palindromes and Watson-Crick conjugates of a word. The number of palindromes in the conjugacy class of a word is studied in \cite{2015arXiv150309112G}. We  investigate the set of $\theta$-conjugates of a word. We study the number of Watson-Crick palindromes in a conjugacy class. We then consider the number of palindromes and Watson-Crick palindromes in the Watson-Crick conjugacy set of a given word.

The paper is organised as follows. In Section \ref{sec-3}, we study the properties of the set of $\theta$-conjugates of a word. We first show that for a given word $w$, the maximum number of elements in the $\theta$-conjugacy of a word is $|w|+1$, and we also characterize the words that attain this maximum number. In Section \ref{sec-4}, we study the distribution of $\theta$-palindromes in the conjugacy class of a word. We show that the conjugacy class of a word can contain at most two distinct palindromes. In Section \ref{sec-5}, we study the number of palindromes in the set of $\theta$-conjugates of the word. We find the structure of the words which have at least one palindrome among its $\theta$-conjugates. Lastly, in  Section \ref{sec-5}, we analyse the number of $\theta$-palindromes in the set of  $\theta$-conjugates of a word. We find the structure of the words which have at least one $\theta$-palindrome among its $\theta$-conjugates. We end the paper with some concluding remarks.

\section{Basic definitions and notations}\label{sec2}
An alphabet $\Sigma$ is a finite non-empty set of symbols. A word over $\Sigma$ is defined to be a finite sequence of symbols from $\Sigma$. $\Sigma^{*}$ denotes the set of all words over $\Sigma$ including the empty word $\lambda$ and $\Sigma^{+}=\Sigma^* \setminus \lambda$. The length of a word  $w \in \Sigma^{*}$ is the number of symbols in a word and is denoted by $|w|$. The reversal of $w=a_{1}a_{2} \cdots a_{n}$ is defined to be a string $w^{R}=a_{n} \cdots a_{2} a_{1}$ where $a_{i} \in \Sigma$. $Alph(w)$ denotes the set of all sub-words of $w$ of length $1$. A word $w$ is said to be a palindrome if $w=w^{R}$.

A word $w\in \Sigma^+$ is called \textit{primitive} if $w = u^i$ implies $w = u$ and $i = 1$.
Let $Q$ denote the set of all primitive words. 
For every word $w \in \Sigma^+$, there exists a unique word $\rho(w)\in \Sigma^+$, called  the {\it primitive root} of $w$,  such that  $\rho(w) \in Q$ and $w = \rho(w)^n $ for some $n \geq 1$. A function  $\theta:\Sigma^{*} \rightarrow \Sigma^{*}$ is said to be an \textit{antimorphism} if $\theta(uv)=\theta(v) \theta(u)$. The function $\theta$ is called an \textit{involution} if $\theta^{2}$ is an identity on $\Sigma^{*}$. 

A word $u \in \Sigma^*$ is said to be a factor of $w$ if $w=xuy$ where $x, y \in \Sigma^*$. If $x= \lambda$, then $u$ is a prefix of $w$ and if $y = \lambda$, then $u$ is a suffix of $w$. A word $u \in \Sigma^*$ is a conjugate of $w \in \Sigma^*$ if there exists $v \in \Sigma^*$  such that
$uv = vw$. The set of all conjugates of $w$, denoted as $C(w)$, is the conjugacy class of $w$.
 A word $u$ is a $\theta$-conjugate of another word $w$ if $uv = \theta(v)w$ for some $v \in \Sigma^*$. The set of all $\theta$-conjugates of $w$ is denoted by $C_{\theta}(w)$.
 For an antimorphic involution $\theta$, a finite word $w$ is called a $\theta$-palindrome if $w = \theta(w)$.
Consider $\Sigma=\{a,b\}$ and an antimorphic involution $\theta$ such that $\theta(a)=b$ and $\theta(b)=a$. Then, the word $abab$ is a $\theta$-palindrome but not a palindrome. For all other concepts in formal language theory and combinatorics on words, the reader is referred to \cite{Hopcroft69,Lothaire97}.

Throughout the paper, we take $\theta$ to be an antimorphic involution over $\Sigma$.
\section{Conjugacy and theta-Conjugacy  of a word.}\label{sec-3}
In this section, we study the conjugacy class and the $\theta$-conjugacy set $C_\theta(w)$ for a word $w$. We show some  general properties of  the set $C_\theta(w)$. We also characterize words for which $|C_\theta(w)|=|w|+1$ which is the maximum number of  of $\theta$-conjugates that a word of length $|w|$ can have.

We recall the following result from \cite{watson} for $\theta$-conjugates of a word.
\begin{proposition}\label{k1}
Let $u$ be a $\theta$-conjugate of $w$. Then, for an antimorphic involution
$\theta$, there exists $x, y \in  \Sigma^*$ such that either $u = xy$ and $w = y\theta(x)$ or $w = \theta(u)$.
\end{proposition}
Thus, we can deduce that for a word $w$,
$$ C_{\theta}(w) = \{ \theta(v)u~:~ w=uv, ~u,v\in \Sigma^*\}$$
Also, Proposition \ref{k1} implies that the maximum number of elements in  $C_{\theta}(w)$ is $|w|+1$. It is also clear that if $w$ is a $\theta$-palindrome, then the maximum number of elements in  $C_{\theta}(w)$ is $|w|$. We illustrate the concept of $\theta$-conjugacy of a word with the help of an example and show that the number of elements in the $\theta$-conjugacy of a word $w$ may or may not reach the maximum.
\begin{example}\label{o1} Consider $\Sigma=\{a,b,c\}$ and $\theta$ such that $\theta(a)=b$, $\theta(b)=a,\;\theta(c)=c$. 
\begin{enumerate}
 \item If $w = aac$, then $C_{\theta}(w) = \{ aac, caa, cba, cbb \}$ and $|C_{\theta}(w)| = 4 = |w|+1$. Note that $aac$ is a primitive word that is neither a palindrome nor a $\theta$-palindrome.
    \item If $w = abb$, then $C_{\theta}(w) = \{ abb, aab, aaa \}$ and $|C_{\theta}(w)| = 3 < |w|+1 = 4$. Note that $abb$ is a primitive word that is neither a palindrome nor a $\theta$-palindrome.
       \item If $w = bccb$, then $C_{\theta}(w) = \{ bccb, abcc, acbc, accb, acca \}$ and $|C_{\theta}(w)| = 5 =|w|+1$. Note that $bccb$ is a palindrome.
    \item If $w = aba$, then $C_{\theta}(w) = \{ aba, bab, baa\}$ and $|C_{\theta}(w)| =3 <|w|+1$.  Note that $w$ is a palindrome.
    \item  If $w=ab$, then  $C_{\theta}(ab) = \{ ab, aa \}$ and  $|C_{\theta}(w)| = 2 =|w|$. Note that $w$ is a $\theta$-palindrome.
    \item 
    If $w = abcab$, then $C_{\theta}(w) = \{ abcab, aabca, ababc, abcaa \}$ and $|C_{\theta}(w)| = 4<|w|$.  Note that $w$ is a $\theta$-palindrome.
\item If $w=aaa$, then $C_{\theta}(w) = \{ aaa, baa, bba, bbb \}$ and  $|C_{\theta}(w)| = |w|+1 = 4$. Note that $w$ is not a primitive word but $|C_{\theta}(w)| = |w|+1$.
\end{enumerate}
\end{example}
 It is well known that the maximum number of elements in a conjugacy class of a word $w$ is $|w|$ and it is attained if $w$ is primitive. This is not true in general for $\theta$-conjugacy of $w$. It is clear from Example \ref{o1} that there are primitive words with the maximum number of elements in the set $C_{\theta}(w)$ to be less than $|w|$, equal to $|w|$ and equal to $|w|+1$.
It is also clear from Example \ref{o1} that the maximum number of elements in the set $C_{\theta}(w)$, i.e., $|w|+1$, can be attained by both primitive as well as a non-primitive word.

We now characterize words $w$ with exactly $|w|+1$ elements in $C_{\theta}(w)$. We first recall some general results from \cite{Schutz62}.
\begin{lemma} \label{4}
Let $u,v,w \in \Sigma^{+}$.
\begin{itemize}
     \item If $uv=vu$ then, $u$ and $v$ are powers of a common primitive word. 
     \item If $uv=vw$ then, for $k \geq 0$, $x \in \Sigma^{+}$ and $y \in \Sigma^{*}$,$u=xy$, $v=(xy)^{k}x$, $w=yx$.
 \end{itemize}
\end{lemma}
We have the following result.
\begin{proposition}\label{b1}
Let $w \in \Sigma^*$. Then, $|C_{\theta}(w)| =|w|+1$ iff $w$ is not of the form $(\alpha \beta)^{i+1}\alpha v$ where $\alpha, v \in \Sigma^*$, $\beta \in \Sigma^+$, $i \geq 0$ and $\alpha,\; \beta$ are $\theta$-palindromes. 
\end{proposition}
\begin{proof}
Let $w \in \Sigma^*$. We prove that  $|C_{\theta}(w)| < |w|+1$ iff $w=(\alpha \beta)^{i+1}\alpha v$ where $\alpha, v \in \Sigma^*$, $\beta \in \Sigma^+$, $i \geq 0$ and and $\alpha,\; \beta$ are $\theta$-palindromes.\\
Let $|C_{\theta}(w)| < |w|+1$, then at least two $\theta$-conjugates of $w$ are equal. Let $\theta(v)u$ and $\theta(y)x$ be the two $\theta$-conjugates of $w$ that are equal where $u, v, x, y \in \Sigma^*$ for $w= uv = xy$. Without loss of generality, let us assume that $|u|>|x|$. Then, $|v|<|y|$, $ u = x y_1$  and $y = y_1 v$ for some $y_1 \in \Sigma^+$.
Now, $$\theta(v)u = \theta(y)x\implies  \theta(v)xy_1 = \theta(y_1 v)x \implies \theta(v)xy_1 = \theta(v) \theta(y_1)x  \implies xy_1 = \theta(y_1)x.$$ Then by Lemma \ref{4}, $\theta(y_1) = \alpha \beta$, $x = (\alpha \beta)^i\alpha$ where $\beta \in \Sigma^+$, $\alpha \in \Sigma^*, i \geq 0$ such that $\alpha$ and $\beta$ are $\theta$-palindromes. 
Hence, $w = uv =xy= (\alpha \beta)^{i+1}\alpha v $ and $\alpha,\; \beta$ are $\theta$-palindromes. 
\\
Conversely, let $w = (\alpha \beta)^{i+1}\alpha v$ where $\alpha, v \in \Sigma^*$, $\beta \in \Sigma^+$, $i \geq 0$ and $\alpha,\; \beta$ are $\theta$-palindromes. Now, $w' = \theta(v) (\alpha \beta)^{i+1}\alpha \in C_{\theta}(w)$ and  $w'' =  \theta(\beta \alpha v) (\alpha \beta)^i \alpha \in C_{\theta}(w)$. Consider $$w'' = \theta(v) \theta(\beta \alpha) (\alpha \beta)^i \alpha  = \theta(v) \alpha \beta (\alpha \beta)^i \alpha = \theta(v) (\alpha \beta)^{i+1} \alpha = w'.$$ Therefore, $|C_{\theta}(w)| < |w|+1$.
\end{proof}

The conjugacy  operation on words is an equivalence relation. However,  the $\theta$-conjugacy on words is not an equivalence relation. Note that from Example \ref{o1}, it is clear that $aaa$ is a $\theta$-conjugate of $abb$, but $abb$ is not a $\theta$-conjugate of $aaa$.  Thus, $\theta$-conjugacy on words is not a
symmetric relation, and hence, not an
equivalence relation.\\ \\
We recall the following from \cite{Shyr}.
\begin{lemma}\label{2}
Let $u = z^{i}$ for a primitive word $z$ over $\Sigma$. Then, the  conjugacy class of $u$  contains exactly  $|z|$  words.
\end{lemma} 
We show that Lemma \ref{2} do not hold for $\theta$-conjugacy. Infact we show that for any word $z$, the number of $\theta$-conjugates in  $z^i$ may not be equal to that of  the number of $\theta$-conjugates in the word $z$. We illustrate with the help of  an example.
\begin{example}
Let $\Sigma=\{a, b,c\}$, and  $\theta$ such that $\theta(a)=b, \theta(b)=a, \theta(c)=c$. Then,
\begin{enumerate}
\item $C_{\theta}(ac) =  \{ ac, ca, cb \}$;
    \item $C_{\theta}((ac)^2) =  \{ acac, caca, cbac, cbca, cbcb \}$; \item $C_{\theta}((ac)^{3}) =  \{ acacac, cacaca, cbacac, cbcaca, cbcbac, cbcbca, cbcbcb \}.$ 
\end{enumerate} 
\end{example}
 
We show that the number of $\theta$-conjugates of a word $z^i$ are greater than the number of  $\theta$-conjugates of a word $z^j$ for $i>j$ and $z\in \Sigma^*$ and $|C_{\theta}(z)| \neq 1$. We have the following result.
\begin{lemma}\label{gg5}
Let $z \in \Sigma^*$. If $|C_{\theta}(z)| \neq 1$, then  $|C_{\theta}(z)| < |C_{\theta}(z^2)|$.
\end{lemma}
\begin{proof}
Let $z\in\Sigma^*$ such that $|C_{\theta}(z)| \neq 1$. Let $w=\theta(v)u \in C_{\theta}(z)$ such that $z=uv$.  Now $z^2 = zz = uz'v$ where $z'=vu$. Then, $w_1 = \theta(v)uz' \in C_{\theta}(z^2)$. So for each element $w$ in $C_{\theta}(z)$, there exist an element $w_1$ in $C_{\theta}(z^2)$ such that $w$ is a prefix of $w_1$. Let $z = u_1v_1$. Then $w' = \theta(v_1)u_1 \in C_{\theta}(z)$. So there exist an element $w_1'= \theta(v_1)u_1z'' $ in $C_{\theta}(z^2)$ where $z''= v_1u_1$. Now if $w \neq w'$ then $w_1 \neq w_1'$.  Hence,  $|C_{\theta}(z)| \leq |C_{\theta}(z^2)|$.\\
Now $\theta(uv) \in C_{\theta}(z)$. Note that the words  $\alpha_1=\theta(vuv)u$ and $\alpha_2=\theta(uv)uv\in C_{\theta}(z^2)$ have $\theta(uv)$ as their prefix. If $\alpha_1 =\alpha_2$,  we get $\theta(v)u = uv$ such that $z=uv$ for all $u,v \in \Sigma^*$.
  Thus, we get $\theta(v)u=uv$ for all $u, v \in \Sigma^*$, so $C_{\theta}(z) = \{ uv \}$, i.e., $|C_{\theta}(z)|=1$  which is a contradiction.
  Therefore, there exist at least two distinct elements $\alpha_1$ and $\alpha_2$ in $C_{\theta}(z^2)$ whose prefix is $\theta(uv)$. Thus  $|C_{\theta}(z)| < |C_{\theta}(z^2)|$.
\end{proof}
We deduce an immediate result.
\begin{corollary}\label{gg3}
Let $z \in \Sigma^*$. If $|C_{\theta}(z)|\neq1$,  then 
$|C_{\theta}(z^i)| < |C_{\theta}(z^{i+1})|$ for $i\geq 1 $.
\end{corollary}

We recall the following from \cite{Kari2010}.
\begin{proposition}\label{gg1} If $uv = \theta(v)u$ and
 $\theta$ is an antimorphic involution, then $u = x(yx)^i$, $v = yx$ where $i \geq 0$ and $u, x, y$ are $\theta$-palindromes,
where $x\in \Sigma^*,\;y \in \Sigma^+$.
\end{proposition}
\begin{lemma}\label{gg4}
Let $z\in \Sigma^*$  then $|C_{\theta}(z)|=1$ iff $z=a^n$ such that $\theta(a)=a$ for  $a\in \Sigma$.
\end{lemma}
\begin{proof}
Let $|C_{\theta}(z)|=1$, then $z$ is a $\theta$-palindrome. So, $z$ is of the form $uv\theta(u)$ where $v$ is a $\theta$-palindrome. Let  $u=au'$ where $a\in \Sigma$. As $|C_{\theta}(z)|=1$, we have, $z=au'v\theta(u')\theta(a)=aau'v\theta(u')$. Then, $u'v\theta(u')\theta(a)=au'v\theta(u')$. Take $u'v\theta(u')=u_1$, then $u_1\theta(a)=au_1$. By Proposition \ref{gg1}, we have $\theta(u_1)=x(yx)^i$ and $\theta(a)=yx$ where $x\in \Sigma^*,\;y \in \Sigma^+$. Then $x=\lambda$ and as $y$ is a $\theta$-palindrome, $u_1=a^i$ and $\theta(a)=a$. Hence, $z=a^{i+1}$. Converse is straightforward.
\end{proof}

Hence, we deduce the following by  Corollary \ref{gg3} and Lemma \ref{gg4}.
\begin{theorem}
Let $z \in \Sigma^*$ such that $z\neq a^n$ for $a\in \Sigma$ such that $\theta(a)=a$, then 
$|C_{\theta}(z^i)| < |C_{\theta}(z^{i+1})|$ for $i\geq 1 $.
\end{theorem}

 \section{Theta palindromes in the conjugacy class of a word}\label{sec-4}


The distribution of palindromes in the conjugacy class of a word was studied in \cite{2015arXiv150309112G}. The authors proved that there are at most two distinct palindromes in the conjugacy class of a given word. In this section, we show an analogous result pertaining to the distribution of $\theta$-palindromes and prove that the conjugacy class of any given word also contains at most two $\theta$-palindromes. We also provide the structure of such words.
We begin by recalling the following from \cite{2015arXiv150309112G}.
\begin{lemma}
If the conjugacy class of a word contains two distinct
palindromes, say $uv$ and $vu$, then there exists a word $x$ and a number $i$ such that $xx^R$
is primitive, $uv = (xx^R)^{i}$, and $vu = (x^R x)^{i}$.
\end{lemma}
In this section, we study the distribution of $\theta$-palindromes in the conjugacy class of a word.

We first observe the following.
\begin{proposition}\label{3}
If $w^{ n_1}$ is a $\theta$-palindrome, then $w^{n_2}$ is a $\theta$-palindrome for $n_1,n_2\geq 1$.
\end{proposition}

We recall the following from \cite{2015arXiv150309112G}.
\begin{lemma}
If $u\neq u^R$ and $uu^R = z^{i}$ for a primitive word $z$ then, $i$ is odd and $z$ = $xx^R$ for some $x$.
\end{lemma}

We deduce the following.
\begin{lemma}\label{1}
Suppose $u\neq \theta(u)$ and $u\theta(u) = z^{i}$ for a primitive word $z$. Then, $i$ is odd and $z$ = $x\theta(x)$ for some $x$.
\end{lemma}
\begin{proof}
If $i$ is even, then $u\theta(u) = (z^{\frac{i}{2}})^2$. Hence, $u = \theta(u)$, contradicting the conditions of the lemma. So $i$ is odd and then $|z|$ is even. Let $z = xx'$, where $|x| = |x'|$.
We see that $x$ is a prefix of $u$ and $x'$
is a suffix of $\theta(u)$. Hence, $x' = \theta(x)$, as required.
\end{proof}
\begin{lemma}\label{6}
If the conjugacy class of a word contains two distinct
$\theta$-palindromes, say $uv$ and $vu$, then there exists a word $x$ and a number $i$ such that $x\theta(x)$
is primitive, $uv = (x\theta(x))^{i}$, and $vu = (\theta(x) x)^{i}$.
\end{lemma}

\begin{proof}
We use induction on $n=|uv|$.
For $n=2$, if there are two distinct  $\theta$-palindromes, then they must be of the form $x\theta(x)$ and $\theta(x)x$ where $x\in \Sigma$ and $x\neq \theta(x)$, i.e.,  $x\theta(x)$ is primitive.
For the inductive step, assume $|u|\geq |v|$ without loss of generality.
If $|u|=|v|$, then $v = \theta(u)$. By Lemma \ref{1}, we get $uv = (x\theta(x))^{i},\; vu = (\theta(x)x)^{i}$ for a  primitive
word $x\theta(x)$ and $i \geq 1$.\\
Now let $|u|>|v|$. Then $v$ is a prefix and suffix of $\theta(u)$.
$\theta(u)=vw_1=w_2v$ for $w_1,w_2\in \Sigma^*$. This implies by Lemma \ref{4}, we obtain $v = (st)^{i}s$, and hence, $\theta(u) = (st)^{(i+1)}s$,
 for $s\neq \lambda$, and $i \geq 0$. 
Looking at the central factor of the  palindromes
$$uv = (\theta(s)\theta(t))^{(i+1)}\theta(s)(st)^{i}s$$ 
$$vu= (st)^{i} s(\theta(s)\theta(t))^{(i+1)}\theta(s),$$
we see that $st$ and $ts$ are also $\theta$-palindromes. If $t = \lambda$, then $s$ is a $\theta$-palindrome,
implying $uv = vu$, which is a contradiction as $uv$ and $vu$ are distinct. 
If $st = ts$, then by Lemma \ref{4}, both $s$ and $t$ are powers of some primitive word $z$ and  by Proposition \ref{3}, $s, t$, and $z$ are $\theta$-palindromes, which
again implie $uv = vu$ which is a contradiction. Thus, $st\neq ts$ with $|st| < n$ and by inductive
hypothesis we get $st = (x\theta(x))^{j}$, $t s = (\theta(x)x)^{j}$
for some
primitive word $x\theta(x)$. Then, we have 
\begin{equation}\label{e1}
v = (x\theta(x))^{ji} s = s(\theta(x)x)^{ji}
\end{equation}


Note that, if $s = (x\theta(x))^{k}x_1$ where $x=x_1x_2$, then Eq \ref{e1} becomes
$$(x_1x_2\theta(x_1x_2))^{j(i+1)}(x\theta(x))^{k}x_1  = (x\theta(x))^{k}x_1(\theta(x_1x_2)x_1x_2)^{j(i+1)}$$
Then by the comparing suffix,  $x_2\theta(x)x_1=\theta(x)x_1x_2$ which is a contradiction, since all conjugates of $x\theta(x)$ are distinct by Lemma \ref{2}. The argument for the case when $s = (x\theta(x))^{k}x\theta(x_2)$ is similar. If $s = x(\theta(x)^k$, then Eq \ref{e1} becomes
$$ (x\theta(x))^{ji} (x\theta(x))^k = (x\theta(x))^k(\theta(x)x)^{ji}$$
and by comparing the suffix,  we get $x\theta(x) = \theta(x)x$ which is a contradiction as $x\theta(x)$ is primitive.
Thus, 
$s = (x\theta( x))^{k} x$ for some $k, 0 \leq k < j$. Then we can easily compute $t, \theta(s)$, and $\theta(t)$ by using $s$ (we know the value of $st$) to
get $uv = (x\theta(x))^{2j(i+1)}, \; vu = (\theta(x)x)^{2j(i+1)}$. Hence, the proof.
\end{proof}
We give examples of  words that have zero, one and two $\theta$-palindromes each in their conjugacy class.
\begin{example}\label{ex1}
Let $\Sigma=\{a,b,c\}$, and  $\theta$ such that $\theta(a)=b,\;\theta(b)=a$ and $\theta(c)=c$. Then
\begin{enumerate}
    \item The word $aaa$ has zero $\theta$-palindromes in its  conjugacy class.
    \item  The word $cabab$ has exactly one $\theta$-palindrome $abcab$ in its  conjugacy class.
    \item  The word $abab$ has two exactly two $\theta$-palindromes $abab$ and  $baba$ in its  conjugacy class. \end{enumerate}

\end{example}

We know from Example \ref{ex1} that there exists words that contain exactly, zero, one or two $\theta$-palindromes in their conjugacy class. However, in the following result, we show that a conjugacy class of any word contains at most two distinct $\theta$-palindromes and we also find the structure of words that contain exactly two $\theta$-palindromes in their conjugacy class.

\begin{theorem}
The conjugacy class of a word  contains at most two $\theta$-palindromes. It has exactly two $\theta$-palindromes iff it contains a  word of the form $(\alpha  \theta(\alpha))^{l}$, where $\alpha \theta(\alpha)$ is  primitive  and $l\geq 1$. Note that such a word  has even length.
\end{theorem}
\begin{proof}
It is clear from Example \ref{e1} that a word can have upto $2$ $\theta$-palindromes.  WLOG, we may assume that $w$ is primitive.  Suppose there are  two $\theta$-palindromic conjugates of a word $w$, then by Lemma \ref{6}, they are  of the form $uv=x\theta(x)$ and $vu=\theta(x)x$ where $x\theta(x)$ is primitive. We  show that the conjugacy class of $w$ contains no other $\theta$-palindromes. 

Consider the conjugacy class of $w = x\theta(x)$.  Let $u_1$ and  $u_2$ be such that $u_1u_2 = x\theta(x)$ and  $u_2 u_1$
is a $\theta$-palindrome. If $|u_1|\neq |u_2|$, say, $|u_1|<|u_2|$, then we apply the same argument in Lemma \ref{6} and obtain $u_1u_2 = (y \theta(y))^{2k}$
for some $y$ and
$k$. But this is impossible, because $u_1 u_2 = x\theta(x)$ is primitive. Hence, $u_1=u_2$, and
$u_2u_1 = \theta(x)x$. Thus, the conjugacy class contains exactly two $\theta$-palindromes and the words are of the structure described in Lemma \ref{6}.
\end{proof}

 \section{Palindromes in the set of all Theta-conjugates of a word}\label{sec-5}
 The concept of $\theta$-conjugacy of a word was introduced in \cite{watson} to incorporate the notion of Watson-Crick involution map to the conjugacy relation. In this section, we count the number of distinct palindromes  in  $C_{\theta}(w)$, $w\in \Sigma^*$. 
We find the structure of words which have at least one palindrome in the set of their $\theta$-conjugates. We also show that if a word is a palindrome, then there can be at most two palindromes among its $\theta$-conjugates.

It is clear from Example \ref{o1}, that the words $acc$,\;$abb$ and $aaa$ have zero, one and two palindromes, respectively, in their respective conjugacy classes. We now find the structure of words with at least one palindrome among their $\theta$ conjugates.

\begin{theorem}
Given $w \in \Sigma^*$, $C_{\theta}(w)$   contains at least one palindrome iff $w = u \theta(x^R) x $ or $w = y v \theta(y^R)$ where $u, v, x, y \in \Sigma^*$  and $u, v$ are palindromes. 
\end{theorem}
\begin{proof}
We first show the converse. Let  $w = u \theta(x^R) x$, for a palindrome $u$ such that $u,x \in \Sigma^{*}$.  Then, $\theta(x) u \theta(x^R) \in C_{\theta}(w)$ is a palindrome. Similarly, for $w = y v \theta(y^R)$ where $v, y \in \Sigma^*$  and $v$ is a palindrome, $y^R \theta(v) y \in C_{\theta}(w)$ is a palindrome. Therefore,  when $w = u \theta(x^R) x $ or $w = y v \theta(y^R)$ for palindromes $u$ and $v$, $C_{\theta}(w)$ contains at least one palindrome.
Now, let there exist at least one palindrome in $C_{\theta}(w)$. Let $w = uv$ where $u, v \in \Sigma^*$ such that $\theta(v)u$ is a palindrome in $C_{\theta}(w)$. Then, we have the following cases:
\begin{enumerate}
    \item Let $|u|<|v|$ and let $v = v_1v_2$ such that $|v_2| = |u|$. Now, $\theta(v)u = \theta(v_2) \theta(v_1) u$. Since, $\theta(v)u$ is a palindrome and $|v_2|=|u|$, we get $\theta(v_2) = u^R$ and $\theta(v_1)$ is a palindrome. Then, $$w= uv = u v_1 v_2 = u v_1 \theta(u^R).$$
    Since $\theta(v_1)$ is a palindrome, $v_1$ is a palindrome.
    \item If  $|u|=|v|$ then, $\theta(v) = u^R$, i.e., $v = \theta(u^R)$, and hence, $w = u v = u \theta(u^R)$.
    \item Let $|u|>|v|$ and let $u = u_1u_2$ such that $|u_2| = |v|$. Since, $\theta(v)u = \theta(v)u_1u_2 = u^R\theta(v^R)$ we obtain, $u_2 = \theta(v^R)$ and $u_1$ a palindrome. Then,  $w = uv = u_1 u_2 v = u_1 \theta(v^R) v$. 
\end{enumerate}
Hence, in all cases either $w = u \theta(x^R) x $ or $w = y v \theta(y^R)$ where $u, v, x, y \in \Sigma^*$  and $u, v$ are palindromes. 
\end{proof}
It is evident from the definition of $\theta$-conjugacy of a word that, $C_{\theta}(w)$ contains both $w$ and $\theta(w)$. Hence, if $w$ is a palindrome, then $C_{\theta}(w)$ contains at least two palindromes $w$ and $\theta(w)$. In the following, we show that for a palindrome $w$, $w$ and $\theta(w)$ are the only palindromes in the set of all $\theta$-conjugates of $w$.

\begin{theorem}
For a palindrome $w$, the number of palindromes in $C_{\theta}(w)$ is atmost two and is exactly two if $w\neq \theta(w)$.
\end{theorem}
\begin{proof}
Let $w$ be a palindrome. Then $\theta(w)$ is also a palindrome. Suppose there exists a $w' = \theta(v)u \in C_{\theta}(w)$ where $w = uv$ such that $w'$ is a  palindrome, then  $\theta(v)u = u^R\theta(v^R)$.  We have the following cases.
\begin{itemize}
    \item [Case I:] If $|v| > |u|$, then there exists a  $v' \in \Sigma^+$ such that $ \theta(v')u = \theta(v^R)$. We then have, $\theta(u)v' = v^R$ i.e., $v= v'^R \theta(u^R)$. Now, $w=uv = u v'^R \theta(u^R)$. As $w$ is a palindrome, $u^R = \theta(u^R)$, i.e., $u = \theta(u)$. Then, $w' = \theta(v)u = \theta(v)\theta(u) = \theta(uv) = \theta(w)$.
    \item[Case II:] If $|v| = |u|$, then $\theta(v) =u^R$ and since $w$ is a palindrome, we have, $v=u^R$ which implies $v =\theta(v)$ and $u =\theta(u)$.  Thus, $w = uu^R$ and  $w' = \theta(v)u =\theta(u^R)u \theta(u^R)\theta(u) = \theta(w)$.
    \item[Case III:] If $|u|>|v|$, then $v^R$ is a prefix of $u$ and $\theta(v^R)$ is a suffix of $u$ since both $w$ and $w'$ are palindromes. If $|u|\leq 2|v|$ and since $v^R$ is a prefix of $u$ and $\theta(v^R)$ is a suffix of $u$, then $u = v_2^R v_1^R \theta(v_2^R)$ with $v=v_1v_2$ such that $|v_1| = 2|v|-|u|$ and $v_1^R=\theta(v_1^R)$.  Thus, $\theta(u) = u$, and hence, $\theta(w) = \theta(v) \theta(u) = w'$.
       If $|u|>2|v|$, then $u = v^R \alpha \theta(v^R)$ for some $\alpha \in \Sigma^+$. Since $w = uv$, $w' = \theta(v)u$ and $\theta(w)$ are palindromes, we get
        \begin{equation}\label{oneq}
        w= uv = v^R \alpha \theta(v^R)v = v^R u^R = v^R\theta(v)\alpha^Rv ;  
        \end{equation}
        \begin{equation}\label{tweq} \theta(w) = \theta(v)\theta(u) = \theta(v)v^R\theta(\alpha)\theta(v^R) = \theta(u^R)\theta(v^R) = \theta(v)\theta(\alpha^R)v\theta(v^R).
      \end{equation}
      Equations \ref{oneq} and \ref{tweq} gives, $\alpha\theta(v^R) = \theta(v)\alpha^R$ and $v^R\theta(\alpha) =\theta(\alpha^R)v $, respectively. 
            Since $w'$ is a palindrome, 
        $v^R \alpha = \alpha^R v$   
           and $\theta(w')$ is a palindrome. Now $$\theta(w') = \theta(u)v = v^R \theta(\alpha) \theta(v^R) v = \theta(\alpha^R)v \theta(v^R) v = (\theta(w'))^R = v^R \theta(v) \theta(\alpha^R) v .$$
        Then, we get, $$ \theta(\alpha^R)v \theta(v^R)  =  v^R \theta(v) \theta(\alpha^R).$$ 
            Then by Lemma \ref{4}, $v^R\theta(v) = xy$, $\theta(\alpha^R) = (xy)^ix$ and $v\theta(v^R) = yx$ where $i\geq 0$, $y\in \Sigma^*$ and $x \in \Sigma^+$. Now $v^R\theta(v) = \theta(v\theta(v^R))$, i.e., $\theta(yx) = xy$. So $x$ and $y$ are $\theta$ palindromes and hence, $$\theta(w') =  v^R \theta(\alpha) \theta(v^R) v = \theta(\alpha^R)v \theta(v^R) v = \theta(\alpha^R) yx v = (xy)^ixyxv = x(yx)^{i+1}v.$$
        Then using Equation \ref{oneq}, we get, $$ w = uv = v^R\alpha \theta(v^R)v = v^R\theta(v)\alpha^Rv = xy x(yx)^iv = x(yx)^{i+1}v = \theta(w').$$ So, $\theta(w) = w'$.
    \end{itemize}
     Hence, in all the cases, we are done.
    \end{proof}
    Consider the word $w=uu\theta(u)$ where $u$ is a palindrome but not a $\theta$-palindrome. Then, $uuu, u\theta(u)u \in C_{\theta}(w)$ are palindromes. Moreover the word  $w=u^{2i}\theta(u)^i$,  where $u$ is a palindrome but not a $\theta$-palindrome,  has at least two palindromes.
    It is  evident that there exists a non-palindrome $w$ such that $C_{\theta}(w)$ contains more than one palindrome.

     \section{Theta palindromes in the Theta-conjugacy set of a word}\label{sec-6}
\textcolor{blue}{}In this section, for a given word $w$, we study the number of $\theta$-palindromes in the set $C_{\theta}(w),\;w\in \Sigma^*$.  We find the structure of words which have at least one $\theta$-palindrome in the set of their $\theta$-conjugates. We also show that if a word is a $\theta$-palindrome, then there can be at most one $\theta$-palindrome among its $\theta$-conjugates.\\
We first give examples of  words that have zero and one $\theta$-palindrome in their $C_{\theta}(w)$ .
\begin{example}\label{ex2}
Let $\Sigma=\{a,b\}$, and consider $\theta$ such that $\theta(a)=b$ and $\theta(b)=a$ . Then
\begin{enumerate}
    \item  $C_{\theta}(aaa) = \{ aaa, baa, bba, bbb \}.$ Thus, it has zero $\theta$-palindromes .
    \item  $C_{\theta}(abab) = \{abab,aaba,abaa\}.$ The word $abab$ has exactly one $\theta$-palindrome $abab$.
    \end{enumerate}
\end{example}
We now find the structure of words with at least one palindrome among their $\theta$ conjugates.
\begin{theorem}
Given $w \in \Sigma^*$, $C_{\theta}(w)$ contains at least one $\theta$-palindrome iff $w = uxu$ or $w = xuu$ where $u, x, \in \Sigma^*$  and $x$ is a $\theta$-palindrome.
\end{theorem}
\begin{proof}
Let there exist at least one $\theta$-palindrome in $C_{\theta}(w)$ and let $w = uv$ where $u, v \in \Sigma^*$ such that $\theta(v)u$ is a $\theta$-palindrome in $C_{\theta}(w)$. Then, we have the following cases:
\begin{enumerate}
    \item Let $|u|<|v|$. Let $v = v_1v_2$ such that $|v_2| = |u|$. Now, $\theta(v)u = \theta(v_2) \theta(v_1) u$. Since $\theta(v)u$ is a $\theta$-palindrome and $|v_2|=|u|$, we get $v_2 = u$ and $v_1$ is a $\theta$-palindrome. Then, $ w= uv = u v_1 v_2 = u v_1u$ and $v_1$ is a $\theta$-palindrome.
    \item Let $|u|=|v|$. Since $\theta(v)u$ is a $\theta$-palindrome and $|u|=|v|$, $v = u$. Then, $w = u v = uu.$ 
    \item Let $|u|>|v|$. Let $u = u_1u_2$ such that $|u_2| = |v|$. Now, $\theta(v)u = \theta(v)u_1u_2$. Since $\theta(v)u$ is a $\theta$-palindrome and $|u_2|=|v|$, $u_2 = v$ and $u_1$ is a $\theta$-palindrome. Then,  $w = uv = u_1 u_2 v = u_1v v$  where $u_1$ is a $\theta$-palindrome.
\end{enumerate}
Hence, in all the cases either $w = uxu$ or $w = xuu$ where $u, x, \in \Sigma^*$  and $x$ is a $\theta$-palindrome. \\
Conversely, let $x \in \Sigma^{*}$ be a $\theta$-palindrome. If $w = uxu$ for some $u 
\in \Sigma^*$ then, $\theta(xu) u  \in C_{\theta}(w)$ is a $\theta$-palindrome. Similarly, for $w = xxu$ where $u \in \Sigma^*$, $\theta(u)xu \in C_{\theta}(w)$ is  a $\theta$-palindrome. Therefore,  when $w = uxu $ or $w = xuu$ where $u,  x \in \Sigma^*$  and $x$ is $\theta$-palindrome, $C_{\theta}(w)$ contains at least one $\theta$-palindrome. \end{proof}
Consider $w=uv$ such that $w$ is not a $\theta$-palindrome, but  $C_{\theta}(w)$  has a $\theta$-palindrome say  $\theta(v)u$. Since, the $\theta$ conjugacy relation is not an equivalence relation,  we cannot predict the number of $\theta$-palindromes in $C_{\theta}(uv)$ using the fact the  $\theta(v)u$ is a  $\theta$-palindrome. But, for a  $\theta$-palindrome $w$, we find (Theorem \ref{t1}) the exact number of $\theta$-palindromes in the set $C_{\theta}(w)$.
\begin{theorem}\label{t1}
The set $C_{\theta}(w)$ for a $\theta$-palindrome $w$ has exactly one $\theta$-palindrome which is $w$ itself.
\end{theorem} 
\begin{proof}
We prove the statement by induction on the length of $w$. For a word of length $1$, the case is trivial. For a $\theta$-palindrome of length $2$,   say $a_1a_2$,  for $a_i\in \Sigma$, $C_{\theta}(a_1a_2)$ is $\{a_1a_2,\theta(a_2)a_1\}$. Assume that $a_1a_2$ and $\theta(a_2)a_1$ are distinct $\theta$-palindromes, then $\theta(a_2)a_1=\theta(a_1)a_2$
and $a_1a_2=\theta(a_2)\theta(a_1)$. This implies $a_1=\theta(a_2)=\theta(a_1)=a_2$ which is a contradiction. Hence, $a_1a_2=\theta(a_2)a_1$ and there is only one $\theta$-palindrome in $C_{\theta}(a_1a_2)$. We assume that for a word $\alpha$ of length less than $|w|$, if
there is a $\theta$-palindrome $\beta$
in the $C_{\theta}(\alpha)$, then $\alpha=\beta$. Let $w$ be a $\theta$-palindrome. Suppose there exist a $w'=\theta(v)u\in C_{\theta}(w)$ where $w=uv$ such that $w'$ is a $\theta$-palindrome, then  $\theta(v)u=\theta(u)v$.  We have the following cases. 
\begin{itemize}
 \item [Case I:] If $|u|<|v|$, then as
    $\theta(u)$ is a suffix of $v$ and   $u$ is the suffix of $v$, this implies that $u$ is a $\theta$-palindrome. Thus,  $w=uv=\theta(v)\theta(u)=\theta(v)u=w'$. 
    \item [Case II:] If $|u|=|v|$, then $u=\theta(v)=\theta(u)=v$, we have, $w=uv=\theta(v)u=w'$.
       \item [Case III:] If $|u|>|v|$, then 
    $\theta(v)$ is a prefix of $u$ and $v$ is a suffix of $u$ since both are $\theta$-palindromes. 
    If $|u|\leq 2|v|$ and since $\theta(v)$ is a prefix of $u$ and $v$ is a suffix of $u$, then $u=\theta(v_2)v_1v_2$ with $v=v_1v_2$ such that $|v_1|= 2|v|-|u|$ and $v_1$ is a $\theta$-palindrome. Then
   $\theta(u)=u$, and hence, $\theta(w)=\theta(v)u=w'$.     
If $|u|>2|v|$, then $u=\theta(v)u_1v$ where $u_1\in \Sigma^+.$ As $w$ is a $\theta$-palindrome, $uv=\theta(v)\theta(u)$, this implies $\theta(v)u_1vv=\theta(v)\theta(v)\theta(u_1)v$. We have, $u_1v=\theta(v)\theta(u_1)$.    Also, as $w'$ is a $\theta$-palindrome,   $\theta(v)u=\theta(u)v$, this implies $\theta(v)\theta(v)u_1v= \theta(v)\theta(u_1)vv$.  We have,   $\theta(v)u_1= \theta(u_1)v$. Now,  $u_1v=\theta(v)\theta(u_1)$ and $\theta(v)u_1= \theta(u_1)v$. Then, by induction hypothesis, as $u_1v$ and $\theta(v)u_1$ are both $\theta$-conjugates of $u_1v$ and are both $\theta$-palindromes, we have $u_1v=\theta(v)u_1$, thus, $$w=uv=\theta(v)u_1vv=\theta(v)\theta(v)u_1v=\theta(v)u=w'.$$
    \end{itemize}
In all the cases, we are done.
\end{proof}
   Consider the word $w=uxxuxx$ where $x,u$ are distinct $\theta$-palindromes. Then, $\theta(x)uxxux,\; \theta(uxx)uxx \in C_{\theta}(w)$. Moreover, the word $w'=(u^ix^{2i})^{2i}$ has at least two $\theta$-palindromes in the set $C_{\theta}(w')$.
    It is  evident that there exists words $w$ such that $C_{\theta}(w)$ contains more than one $\theta$ palindrome.

\section{Conclusions}
We have studied the distribution of palindromes and $\theta$-palindromes among the conjugates and $\theta$-conjugates of a word.  We have characterized the words which have the maximum number of $\theta$-conjugates. We have found structures of words which have at least one palindrome and $\theta$-palindrome in the set of their $\theta$-conjugates. We have enumerated
palindromes and $\theta$-palindromes in the set $C_{\theta}(w)$ where $w$ is a palindrome and a $\theta$-palindrome, respectively. The maximum number of  palindromes and $\theta$-palindromes in the set of $\theta$-conjugates of a word is still unknown in general. After checking several examples, we believe that the maximum number of palindromes and $\theta$-palindromes in the set $C_{\theta}(w)$ for a word $w$ is two. We have also given examples of the words that achieve the above bound. One of our immediate future work is to prove this bound.

\end{document}